\documentclass[]{llncs}

\usepackage{subfigure}
\usepackage[dvips]{graphicx}
\input{epsf}
\usepackage{amssymb}

\usepackage{algorithm}
\usepackage{algpseudocode}

\begin{document}
\title{The maximum disjoint paths problem on multi-relations social networks}


\titlerunning{Disjoint Paths on Multi-relations Social Networks}  
%
\author{Bang Ye Wu\inst{1}}
\authorrunning{Bang Ye Wu}   
%
\tocauthor{Bang Ye Wu}
\institute{National Chung Cheng University, ChiaYi, Taiwan 621,
R.O.C.\\
\email{bangye@cs.ccu.edu.tw}}

\maketitle              

\begin{abstract}
Motivated by applications to social network analysis (SNA), we study the problem of finding the maximum number of disjoint uni-color paths in an edge-colored graph. 
We show the NP-hardness and the approximability of the problem, and both approximation and exact algorithms are proposed. Since short paths are much more significant in SNA, we also study the length-bounded version of the problem, in which the lengths of paths are required to be upper bounded by a fixed integer $l$. It is shown that the problem can be solved in polynomial time for $l=3$ and is NP-hard for $l\geq 4$.
We also show that the problem can be approximated with ratio $(l-1)/2+\varepsilon$ in polynomial time for any $\varepsilon >0$. Particularly, for $l=4$, we develop an efficient 2-approximation algorithm.
  
\end{abstract}
{\flushleft\bf Keywords: }algorithm, social network analysis, disjoint paths, approximation algorithm, NP-complete.

\section{Introduction}

A social network is usually modeled by a graph $G=(V,E)$, in which $V$ is the set of actors and $E\subseteq V\times V$ is the binary relation we are interested in. In the terminology of graph theory, $V$ is the node set and $E$ is the edge set. 
The \emph{connectivity}, or \emph{node connectivity}, of two nodes is the minimum number of nodes whose removal separates the two nodes. By Menger's theory, it is equal to the maximum number of disjoint paths between the two nodes and also can be thought of as a simpler form of the maximum flow between them. 
In social network analysis (SNA), connectivity is a basic measurement of information flow between nodes and also used to define cohesion group and centralities \cite{free91,Han05,Was94}. Thus computing the connectivity of two nodes is an important problem in SNA.

When there are more than one kinds of relations, we can model a \emph{multi-relations social network} by a graph with more than one edge sets.
Let $c$ be a positive integer. A $c$-relations social network can be described by $G=(V,\mathcal{E})$, in which $V$ is the set of nodes and 
$\mathcal{E}=\{E_1,E_2,\ldots, E_c\}$ is a collection of $c$ edge sets. 
For $1\leq i\leq c$, $E_i\subseteq V\times V$ represents the $i$-th relation, and we shall say the edges in $E_i$ are of \emph{color} $i$. Note that there may be edges of different colors between one pair of nodes.
For a fixed $c$, a graph is called as a $c$-colors graph if there are at most $c$ colored edge sets, and simply a ``color graph'' if the number of colors is not fixed or need not be specified.  

A path is of uni-color if all the edges of the path are of the same color.
Two paths are \emph{internally disjoint} if they have no common internal node, and a set of paths are internally disjoint if they are mutually internally disjoint. In this paper we shall simply use ``disjoint''.
The decision version of the main problem discussed in this paper is defined as follows.
\begin{quote}
{\sc Problem}: The disjoint paths problem on color graphs (CDP)\\
{\sc Instance}: A color graph $G$, two nodes $s,t\in V$ and a positive integer $p$.\\
{\sc Question}: Are there $p$ disjoint uni-color paths from $s$ to $t$?
\end{quote}
In general, the graph may be directed or undirected but in this paper we only consider undirected graphs.
We shall use the name ``CDP$_{c,p}$'' for the decision problem of which the input is a $c$-colors graph.
The maximization version, denoted by \emph{Max CDP} problem, asks for the maximum number of disjoint uni-color paths between two given nodes, which will be called as their \emph{colored connectivity}.
When there is only one color, the maximum number of disjoint paths, i.e, the traditional connectivity, can be computed in polynomial time by solving the maximum flow problem.
But the colored connectivity problem, to our best knowledge, has not been studied yet. 
A related but different problem studied in the literature is the minimum color path problem  which is motivated by communication reliability and the goal is to find a path or two disjoint paths with minimum number of colors \cite{moh00,yuan05}. Other related problems also includes the minimum color-cost path problem \cite{has07} and properly colored path problems, seeing \cite{gut09} for example.

The motivation of studying the colored connectivity is natural. Most of the researches in SNA consider only single relation. But in practical there are more than one kinds of relations. The Max CDP problem arises if the information flow or the influence spread only along relations of the same kind. Computer virus spreading is an example. One virus usually spreads only along one or several particular computer softwares. 
Conversations among people is another example. People usually talk different topics with the ones of different relations.
Disjoint paths also play an important role in data communication when security or traffic congestion is concerned. Thus the scenario of the Max CDP problem may also occur if different types of links between nodes are considered, either due to different media or different protocols. 

The results and the organization of this paper are as follows.
In Section 2, first we show that the CDP problem is NP-complete even for 2-colors graphs and that the Max CDP problem cannot be approximated with ratio less than two, unless NP=P. And then we give an $O(mn)$-time $c$-approximation algorithm for $c$-colors graphs. Throughout this paper, $m$ and $n$ denote the numbers of edges and nodes of the input graph $G$, respectively. An extreme example is given to show the tightness of the ratio.
Also we give an $O((m+n)c^{n})$ time exact algorithm for the problem.
Since, in social network analysis, short paths are considered much more significant than long paths, we also study the length-bounded version of the Max CDP problem, namely $l$-LCDP, in which the lengths of solution paths are required to be upper bounded by a fixed integer $l$.
In Section 3, we show that the $l$-LCDP problem can be solved by graph matching for $l=3$ and is NP-hard for $l\geq 4$. We also show that, 
for any fixed $\varepsilon>0$, the $l$-LCDP problem can be approximated with ratio $(l-1)/2+\varepsilon$ in polynomial time. Particularly, for a $c$-colors graph, we give an efficient 2-approximation for $l=4$ with time complexity $O(p^2(c^2n+cm))$, in which $p$ is the number of paths found by the algorithm. In most of the applications, it is a linear time algorithm. 

\section{Complexity and approximability}

In this section, we show the complexity and the approximability of the CDP problem.
First, in Section 2.1, we show that the problem is NP-complete, and the proof also implies that the Max CDP problem is NP-hard and cannot be approximated with ratio less than two, unless NP=P. In Section 2.2, we give a simple $c$-approximation algorithm for $c$-colors graphs and an extreme example to show the sharpness of the ratio. In Section~\ref{sec:exact} we propose an algorithm for finding the exact solution.
For a $c$-colors graph $G=(V,\{E_1,E_2,\ldots E_c\})$, we shall denote $(V, E_i)$ by $G_i$.

\subsection{NP-completeness}
To show the NP-hardness of the CDP problem, we introduce the following similar problem, named MCDP$_c$ in short. 
\begin{quote}
{\sc Problem}: The multi-pairs disjoint paths problem on $c$-colors graphs \\
{\sc Instance}: A $c$-colors graph $G$, $c$ pairs $(s_i,t_i)$, $1\leq i\leq c$, of nodes.\\
{\sc Question}: Is there a color-$i$ path $P_i$ from $s_i$ to $t_i$ for each $1\leq i\leq c$ such that $P_i$ and $P_j$ are internally disjoint for all $i$ and $j$?
\end{quote}

The reduction from the MCDP$_2$ problem to the CDP$_{2,2}$ problem is quite straightforward.
We first assume that all nodes in the given pairs are distinct, and the other case will be explained later.
For an instance of the MCDP$_2$ problem, we construct a graph $G'$ from $G$ by adding two new nodes $s$ and $t$, as well as four edges $(s,s_1)$, $(s,s_2)$, $(t,t_1)$, and $(t,t_2)$.
The edges $(s,s_1)$ and $(t,t_1)$ have color one and the other two new edges have color two.
Apparently there exist two disjoint uni-color $st$-paths in $G'$ if and only if the answer of the MCDP$_2$ problem is also ``yes''.   
Therefore if the MCDP$_2$ problem is NP-complete, so is the CDP$_{2,2}$ problem.
In the case that $s_1=s_2$, we can add a \emph{duplicate} $s_1'$ of $s_1$ such that $s_1'$ has the same neighbors as $s_1$, and the edges incident to $s$ are $(s,s_1)$ and $(s,s_1')$ instead. Other cases that any two nodes in the given pairs are not distinct can also be handled similarly. 
We shall show the NP-completeness of the MCDP$_2$ problem by transformation from the SAT problem. We remind that the MCDP problem on 1-color graphs is polynomial-time solvable when the number of pairs is fixed \cite{rob95,sey80,shi80}.

\begin{figure}[t]
\begin{center}
\includegraphics[scale=0.8]{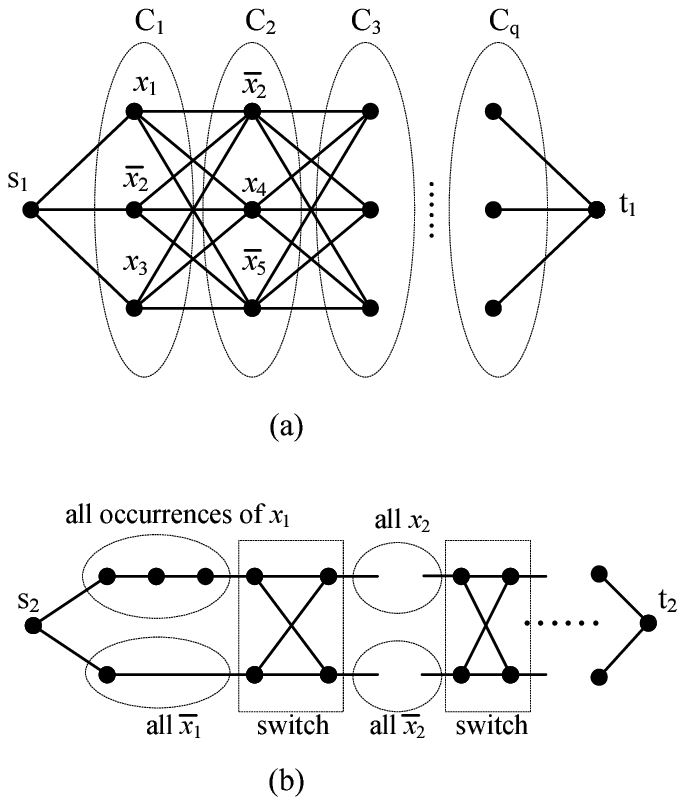}
\caption{Transformation from SAT problem to MCDP$_2$ problem: (a). edges of color 1; and
(b). edges of color 2.}
\label{npcf1}
\end{center}
\vspace{-20pt}
\end{figure}
Let $C_i$, $1\leq i\leq q$ be
the clauses of the SAT problem and $x_i$, $1\leq i\leq r$, the
variables.
We construct a 2-colors graph $G=(V,\{E_1,E_2\})$ as follows.
The node set mainly consists of $\{s_1,t_1,s_2,t_2\}\cup\{x_i^j|x_i\in C_j, 1\leq i\leq
r, 1\leq j\leq q\}\cup \{\bar{x}_i^j|\bar{x}_i\in C_j, 1\leq i\leq
r, 1\leq j\leq q\}$, and some other nodes for some ``switches'' (explained later). 
The edges of color 1 and 2 are depicted in Figure~\ref{npcf1}.

$G_1$ is an $(q+2)$-stages graph, in which the $i$-th stage corresponding to a clause $C_i$ for $1\leq i\leq q$, and the $0$-th and the $(q+1)$-th stages are $s_1$ and $t_1$, respectively. 
Two consecutive stages are connected as a complete bipartite graph.
Note that, for simplicity, the super scripts of nodes are not shown in the figure. 
Different nodes are used to represent a same literal $x_i$ or $\bar{x}_i$ appearing in different clauses.

For color 2, all occurrences of a same literal, i.e, $x_i^j$ or $\bar{x}_i^j$ for all $j$, are connected to form a path, and the four paths of two consecutive variables are connected by a $2\times 2$ switch as shown in the figure.

\begin{lemma}\label{np1}
If and only if there is a truth assignment satisfying all the clauses, there are
an $s_1t_1$-path in $G_1$ and an $s_2t_2$-path in $G_2$, which are disjoint.
\end{lemma}
\begin{proof}
If the instance of SAT problem is satisfiable, we may have an $s_2t_2$-path in $G_2$ passing through all literals which are assigned False. That is, for each $i$, the path passes through $x_i$ if $x_i=$False; and through $\bar{x}_i$ otherwise. Since this truth assignment satisfies all clauses, each clause has a literal assigned True, and therefore there is a path from $s_1$ to $t_1$ in $G_1$.

Conversely, suppose that there are two such disjoint paths. 
Since there is an $s_1t_1$-path in $G_1$, each stage has a node not used by the path in $G_2$.
We observe that, in $G_2$, any $s_2t_2$-path passes through all occurrences of either $x_i$ or $\bar{x}_i$ for every $i$. Therefore if we assign $x_i$ True if it is not passed by the path in $G_2$ and assign False otherwise, every clause has a literal assigned True and the instance is satisfiable. 
\qed\end{proof}
Since the MCDP$_2$ and the CDP problems are apparently in NP, we obtain the following theorem. 
\begin{theorem}
The MCDP$_2$ problem is NP-complete. The CDP problem is NP-complete even for determining if there exist 2 paths in a 2-colors graph.
\end{theorem}

\begin{corollary}
The Max CDP problem is NP-hard and cannot be approximated in polynomial time with ratio $2-\varepsilon$ for any $\varepsilon>0$, unless NP=P.
\end{corollary}
\begin{proof}
Since determining one or two paths is NP-complete, it is impossible to approximate the optimal  with ratio less than two in polynomial time, unless NP=P.
\qed\end{proof}

\subsection{An approximation algorithm}\label{sec:app}
By $\kappa_i(s,t)$ we denote the connectivity of $s$ and $t$ in graph $G_i$, i.e., the maximum number of disjoint paths between them. When the subscript is omitted, $\kappa(s,t)$ denotes the maximum number of disjoint paths of uni-color.  
We show the following greedy algorithm is a $c$-approximation algorithm for $c$-colors graphs.
\begin{quote}
For each color $i$, find $\kappa_i(s,t)$.
Select the color $i$ with maximum $\kappa_i(s,t)$ and put these paths into solution.
Remove all internal nodes of these paths, and then repeat the previous step until no path remains.	
\end{quote}

\begin{theorem}
The Max CDP problem can be $c$-approximated in $O(mn)$ time for $c$-colors graphs.
\end{theorem}
\begin{proof}
Apparently the optimal solution $\kappa(s,t)\leq c\cdot\max_i \kappa_i(s,t)$. The approximation ratio follows from that the number of paths found by the algorithm is at least $\max_i \kappa_i(s,t)$. The value $\kappa_i(s,t)$, i.e., connectivity in a uni-color graph, can be found by solving a maximum flow problem \cite[p. 212]{hugh90} and therefore takes $O(\kappa_i(s,t)|E_i|)$ time  \cite{cor01,iflow}.
In total the algorithm takes $O(\kappa(s,t)m)$ time, or $O(mn)$ time since $\kappa(s,t)<n$. 
\qed\end{proof}

\begin{figure}[t]
\begin{center}
\includegraphics[scale=0.85]{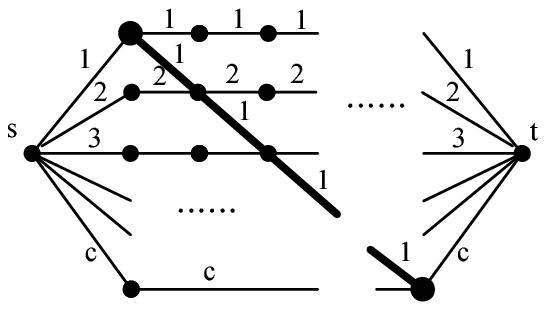}
\caption{A tight example for the $c$-approximation algorithm, in which edges are labeled by their color number.}
\label{fig:tight}
\end{center}
\vspace{-20pt}
\end{figure}
Figure~\ref{fig:tight} illustrates a tight example of the $c$-approximation algorithm. The optimal solution contains $c$ disjoint paths (the horizontal ones), one for each color. But if we choose the bold path of color 1 at the first iteration, the algorithm will find only one path. 

\subsection{An exact algorithm}\label{sec:exact}

First, for any color $i$, if $(s,t)\in E_i$, this path of single edge must be in the optimal solution, and we can put it into the solution and remove this edge. Therefore, in the remaining paragraphs of this paper, we assume $(s,t)\notin E_i$ for any $i$. 
For a $c$-colors graph $G$, define a node coloring $\delta: V-\{s,t\}\mapsto \{1..c\}$.
Two nodes are said to be assigned the same color $i$ if $\delta(u)=\delta(v)=i$.
For the convenience, nodes $s$ and $t$ are thought of having the same color as any node in any coloring.
Let $E_i[\delta]$, $1\leq i\leq c$, denote the subset of $E_i$ in which the two endpoints are assigned the same color $i$ by $\delta$.
Let $E[\delta]=\bigcup_i E_i[\delta]$ and $G[\delta]$ be the uni-color graph induced by the edge set $E[\delta]$.
Suppose that $\mathcal{P}$ is an optimal solution of the Max CDP problem. Let $\delta^*$ be a node coloring such that $\delta^*(v)=i$ if $v$ is on a path of color $i$ in $\mathcal{P}$; and $\delta^*(v)$ is arbitrary otherwise.

We can observe that any path in $\mathcal{P}$ must also be a path in $G[\delta^*]$ and any path in $G[\delta^*]$ corresponds to a uni-color path in $G$. 
Thus, $|\mathcal{P}|$ equals the $st$-connectivities on $G[\delta^*]$ and can be computed in $O(mn)$ time.
If we individually solve the maximum flow problems for all colorings, the total time complexity will be $O(mnc^n)$. By the following observations, the complexity can be reduced to $O((m+n)c^n)$.  
Using the generalized \emph{Gray code}, all the $c^n$ colorings can be arranged in an order $\delta_1, \delta_2,\ldots $ such that two consecutive colorings differ at only one node, and thus the maximum flow corresponding to $G[\delta_{i+1}]$ can be obtained from that corresponding to $G[\delta_{i}]$ by performing at most two breadth-first-searches on the residual graph. 
The next theorem states the result but the detailed proof is omitted here.
\begin{theorem}
There exists an $O((m+n)c^{n})$ time algorithm for the Max CDP problem on $c$-colors graphs.
\end{theorem}

\section{Length-bounded cases}

In this section we discuss the Max CDP problem with bounded length. The length of a path is the number of edges in this path. 
When the path lengths are required to be upper bounded by a fixed integer $l$, we name the problem by $l$-LCDP. 
An edge $(u,v)\in E_i$ will be denoted by $(u,v;i)$, and $(v_1,v_2,\ldots, v_m;i)$ denote a path of color $i$ and visiting $v_1,v_2,\ldots, v_m$ in this order.
The cases of $l\leq 2$ can be easily solved, and we shall discuss the cases of $l=3$ and 4.

\subsection{A polynomial time algorithm for 3-LCDP}
 
The set of all common neighbors of nodes $s$ and $t$ of color $i$ is denoted by $N_{st}^i$.
Recall that we have assumed $(s,t)\notin E_i$ for all $i$, and we need only consider paths of length at least 2.
An $st$-path of length two has the form $(s,v,t;i)$, i.e., any co-neighbor of $s$ and $t$ may contribute a path.
The next claim comes from that any $st$-path of length two is disjoint to any others of length 2 and may intersect at most one $st$-path of longer length.

\begin{claim}
If $v\in N_{st}^i$ for any $i$, there is an optimal solution of the 3-LCDP problem containing the path $(s,v,t;i)$.  
\end{claim}
\vspace{-20pt}
\begin{algorithm}[H]
\caption{}\label{alg:l3}
Input: A $c$-colors graph $G$ and two nodes $s$ and $t$.\\
Output: The maximum number of disjoint uni-color $st$-paths of length at most 3.
\begin{algorithmic}[1]
\State $S\gets\emptyset$; \Comment{solution set}
\For{$k\gets 1$ to $c$}
\State for each node $v\in N_{st}^i$, add path $(s,v,t;i)$ into $S$ and remove $v$ from $G$;
\EndFor
\State $F_i\gets \{<u,v>|\{(s,u),(u,v),(v,t)\}\subseteq E_i\}$ for $1\leq i\leq c$;  \Comment{ordered pairs}
\State $F\gets \bigcup_i F_i$ and construct the directed graph $H$ induced by $F$;
\State find a maximum match $M$ of $H$;
\ForAll{$<u,v>\in M$}
\State add $(s,u,v,t;i)$ into $S$ for some $i$ such that $(u,v)\in F_i$;\label{step:path}
\EndFor 
\State\Return $S$.
\end{algorithmic}
\end{algorithm} 
\vspace{-10pt}
Algorithm 1 is the proposed method for solving the 3-LCDP problem exactly.
Besides the above claim, the correctness of the algorithm is due to the next claim which can be shown by observing that a set of disjoint $st$-paths corresponds to a matching on $H$, and vice versa.  
We remind that defining $F$ as a set of ordered pairs is only for the sake of making step~\ref{step:path} easier. The maximum matching on a directed graph is the same as the one on an undirected graph. 

\begin{claim}
Suppose that $\bigcup_i N_{st}^i=\emptyset$. A maximum matching $M$ of the graph $H$ constructed in Algorithm~\ref{alg:l3} corresponds to an optimal solution of the 3-LCDP problem. 
\end{claim}
The time complexity is dominated by the step of finding a maximum cardinality matching of a general graph, which can be done in $O(\sqrt{n}m)$ time \cite{match}.

\begin{theorem}
The 3-LCDP problem on color graphs can be exactly solved in $O(\sqrt{n}m)$ time.
\end{theorem}

\subsection{The complexity of 4-LCDP and an approximation algorithm}

For the length-bounded case, the notations $\kappa^l_i(s,t)$ and $\kappa^l(s,t)$ are analogous to the ones without superscript but those paths are of length at most $l$. 

\begin{theorem}
The $l$-LCDP problem on $c$-colors graphs is NP-hard for fixed $l\geq 4$ and $c\geq 2$.
\end{theorem}
\begin{proof}
It is sufficient to show the case of $l=4$ and $c=2$.
We show the NP-hardness by transforming from a restrict version of the SAT problem in which there are at most 3 occurrences of each variable. This version of SAT problem still remains NP-complete \cite{npc,tov84}. 
Let $C_i$, $1\leq i\leq q$, be the clauses and $x_j$, $1\leq j\leq r$, the variables. 
For any variable $x_i$, if all the occurrences of $x_i$ are positive, we can assign $x_i$ True and remove $x_i$ from all clauses. The case of all occurrences are negative is similar. Therefore we can assume the occurrences of each variable are neither all positive nor all negative. As a result, both $x_i$ and $\bar{x}_i$ occur at most twice for any $1\leq i\leq r$.  
Given an instance of the restrict SAT problem, we construct a 2-colors graph as in Figure~\ref{npcf2}.
Since the number of occurrences of each literal is at most two, any $st$-path of any color has length at most 4.
\begin{figure}[t]
\begin{center}
\includegraphics[scale=0.85]{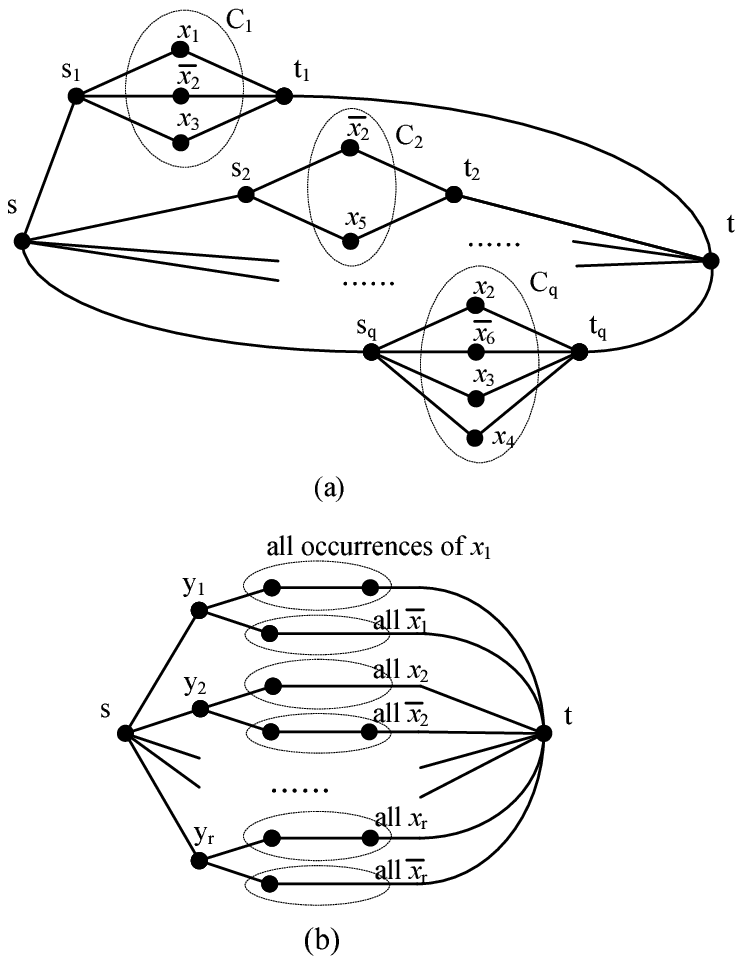}
\caption{Transformation from SAT problem to $4$-LCDP problem: (a). edges of color 1; and
(b). edges of color 2.}
\label{npcf2}
\end{center}
\vspace{-20pt}
\end{figure}
Since both the degree of $s$ and $t$ are $q$, $\kappa_1^4(s,t)\leq q$, and the maximum is achieved if for any clause there is a literal not used in $G_2$. On the other hand, 
$\kappa_2^4(s,t)\leq r$ since the degree of $s$ is $r$. We can also easily find $r$ disjoint $st$-paths in $G_2$ as long as for each $i$ we use either $x_i$ or $\bar{x}_i$ as the internal nodes.
If the SAT instance is satisfiable, let $T$ be a truth assignment satisfying all the clauses.
We choose $x_i$ as internal nodes in $G_2$ if $x_i$ is assigned False in $T$; and $\bar{x}_i$ otherwise.
Then we can have $q$ disjoint $st$-paths of color 1 since there exists a literal assigned True in each clause and thus not used in color 2. The total number of disjoint paths is $q+r$.
Conversely if there are $q+r$ disjoint $st$-paths, there are exactly $q$ paths in $G_1$ and $r$ paths in $G_2$. Therefore for each variable either itself or its negation is used in $G_2$. Since there are $q$ disjoint paths in $G_1$, each clause contains at least one literal not used in $G_2$. So we can assign $x_i$ True if it is not used in $G_2$ and False otherwise, and all the clauses are satisfied. 
\qed\end{proof}

The $c$-approximation algorithm in Section~\ref{sec:app} also works for length bounded case.
We may achieve a better approximation ratio for small $l$.

\begin{theorem}\label{thm:l_app}
For any fixed $\varepsilon>0$ and $l>3$, the $l$-LCDP problem can be approximated with ratio $(l-1)/2+\varepsilon$ in polynomial time.
\end{theorem}

To show Theorem~\ref{thm:l_app}, we introduce the following problem, and Algorithm~\ref{alg:inc} is a $(k/2+\varepsilon)$-approximation algorithm shown in \cite{hur89}.
\begin{quote}
{\sc Problem}: The Maximum Set Packing (MSP) problem\\
{\sc Instance}: A collection $\mathcal{T}$ of $k$-element subsets $T_i$, $1\leq i\leq p$, of a universal set $U$ of total $q$ elements.\\
{\sc Goal}: A maximum disjoint sub-collection of $\mathcal{T}$.
\end{quote}
\vspace{-20pt}
\begin{algorithm}[H]
\caption{}\label{alg:inc}
Input: An instance $\mathcal{T}$ of the MSP problem and an integer parameter $s\geq 1$.\\
Output: A disjoint sub-collection of $\mathcal{T}$.
\begin{algorithmic}
\State $S\gets\emptyset$; \Comment{solution set}
\While{$\exists \, i+1$ disjoint subsets intersecting at most $i$ subsets in $S$ for any $i\leq s$}   
\State replace the $i$ subsets in $S$ with the new $i+1$ subsets;
\EndWhile
\State Output $S$.
\end{algorithmic}
\end{algorithm} 
\vspace{-10pt}
Let OPT denote the maximum number of disjoint subsets and APP denote the result obtained by Algorithm \ref{alg:inc}. It was shown in \cite{hur89} that 
\begin{eqnarray}\label{sdr}
\frac{OPT}{APP}\leq \left\{\begin{array}{ll}
\frac{k(k-1)^r-k}{2(k-1)^r-k}&\mbox{if $s$ is even}\\
\frac{k(k-1)^r-2}{2(k-1)^r-2}&\mbox{if $s$ is odd}
\end{array}\right.
\mbox{where }r= \left\{\begin{array}{ll}
s/2+1&\mbox{if $s$ is even}\\
(s+1)/2&\mbox{if $s$ is odd}
\end{array}\right.
\end{eqnarray}

By transforming to the MSP problem, the $l$-CDP problem can be approximated with ratio $(l-1)/2+\varepsilon$ for any $\varepsilon>0$. 
A direct transformation is as follows.
Let $(G,s,t)$ be an instance of the $l$-LCDP problem.
\begin{itemize}
\item For each uni-color $st$-path of length at most $l$, create a subset $T_i$ consisting of the internal nodes of the path. There are at most $O(n^{l-1})$ subsets and $|T_i|\leq l-1$ for each $T_i$.
\item The elements are all the nodes in the graph except $s$ and $t$.
\item Any disjoint sub-collection corresponds to a set of disjoint uni-color paths.
\end{itemize}
The stop condition of the while-loop can be implemented by enumerating all possible $i+1$ subsets, testing if they are disjoint in $O((i+1)^2l)$ time, and counting the intersected subsets in $S$ in $O((i+1)|S|l)$ time. Since $l$ is fixed, $i\leq s$, and $s$ is also a constant determined by $\varepsilon$, this step takes $O(|\mathcal{T}|^{i+1}\times |S|)$. 
Since $|S|$ is increased at least one after each iteration and bounded by $O(n)$, the naive implementation has time complexity $O(|\mathcal{T}|^{s+1}\times |S|^2)=O(n^{(l-1)(s+1)+2})$, which is polynomial for fixed $l$ and $\varepsilon$. 
Theorem~\ref{thm:l_app} follows from Eq.~(\ref{sdr}), the transformation and the above analysis of the time complexity.

\subsection{An efficient 2-approximation algorithm for 4-LDCP}

Particularly, when $l=4$ and $s=1$, by substituting $k=l-1=3$, the approximation ratio by Eq. (\ref{sdr}) is $(3\times 2-2)/(2^2-2)=2$.
That is, it takes $O(n^8)$ time to compute a 2-approximation of the 4-LDCP problem.
Although in polynomial time, it becomes intractable even for graphs of moderate size.
In the following, we aim at developing a more efficient algorithm for $s=1$ and $l=4$.
Let $S=\{T_i|1\leq i\leq |S|\}$ denote the solution found so far, in which $T_i$ is the set of internal nodes of an $st$-path. Let $V_0=V-\bigcup_i T_i$ be the nodes not used yet.
When $s=1$, the while-condition can be implemented by 
\begin{quote}
	For each $T_i\in S$, determine if there are two disjoint $st$-paths of length at most 4 in $G[V_0\cup T_i]$. 
\end{quote}

The key point is how to determine if $\kappa^4(s,t)\geq 2$ in a color graph without generating all possible paths.
We shall use the following notations. The distance, or shortest path length, between $s$ and $t$ in graph $G_i$ is denoted by $d_i(s,t)$. A node $v$ is an \emph{$st$-cut node} in graph $G_i$ if its removal separates the two nodes, i.e., $\kappa_i(s,t)=0$ after removing $v$. The set of all such cut nodes is denoted by $X_i$.

\begin{algorithm}[H]
\caption{}\label{alg:2path}
Input: A color graph $G$ and two nodes $s$ and $t$.\\
Output: Return True iff there are two disjoint paths of length at most four.
\begin{algorithmic}[1]
\For{each color $i$}
\State remove any node $v$ in $G_i$ such that $d_i(s,v)+d_i(v,t)>4$.\label{st:redu}
\EndFor
\If{$\kappa_i^4(s,t)=2$ for some $i$}
\State\Return True;
\EndIf
\ForAll{$i$ and $j$ such that $\kappa_i^4(s,t)=\kappa_j^4(s,t)=1$}
\If{Test$(i,j)=$True}
\State\Return True;
\EndIf
\EndFor 
\State\Return False.
\end{algorithmic}
\end{algorithm} 
\vspace{-20pt}
\begin{algorithmic}[1]
\Procedure{Test}{$i$,$j$} 
\Comment{testing if there are two disjoint paths in $G_i$ and $G_j$, resp, of length at most four. It is ensured that $X_q(s,t)\neq \emptyset$, for $q=i,j$ and $d_q(s,v)+d_q(v,t)\leq 4$ for any node $v$ in $G_i$ or $G_j$.}
\Repeat
\State $G_i\gets G_i-X_j$; $G_j\gets G_j-X_i$;
\Until{both $G_i$ and $G_j$ are unchanged or $d_q(s,t)>4$ for $q=i$ or $j$;}  

\If{$d_i(s,t)>4$ or $d_j(s,t)>4$ or $X_i\cap X_j\neq \emptyset$}
\State\Return False;\label{step:bigd}
\ElsIf{$d_i(s,t)\leq 3$ or $d_j(s,t)\leq 3$}\label{step:d3}
\State\Return True;
\EndIf
\If{more than two $st$-paths of length 4 in $G_i$ or in $G_j$}
\State\Return True;
\Else\Comment{at most two length-4 paths in $G_i$ and in $G_j$.}
\State determine and return the result by a brute force method;
\EndIf
\EndProcedure
\end{algorithmic}
\begin{lemma}\label{test2}
Algorithm \ref{alg:2path} is correct and takes $O(c^2n+cm)$ time.
\end{lemma}
\begin{proof}
The algorithm returns True iff $\kappa_i^4(s,t)\geq 2$ for some color $i$ or there are two uni-color disjoint paths of two colors. Clearly, what we need to show is the correctness of the procedure Test.

By the assumption that $(s,t)\notin E_i$ for all $i$, we need not consider the case that $d_i(s,t)=1$ or $d_j(s,t)=1$. The test procedure starts with a repeat-until loop to remove any $st$-cut node of one graph from the other. 
Note that the loop is necessary since removing nodes from a graph may result in new cut nodes. But the loop will only be executed at most four times since each graph has one $st$-cut node originally and can have at most three $st$-cut nodes or otherwise $s$ and $t$ will have distance more than 4 (including $\infty$, i.e., disconnected).

Step \ref{step:bigd} deals with the case that the distance between $s$ and $t$ in either graph exceeds 4 or there exists any common $st$-cut node. 
At the beginning of step~\ref{step:d3}, we have that $X_i\cap X_j= \emptyset$ and the distance between $s$ and $t$ at either graph is at least two. 
Let $x\in X_i$.
If $d_i(s,t)=2$, there exists a (unique) $st$-path $(s,x,t)$ of color $i$. Immediately the output should be True since $d_j(s,t)\leq 4$ and $x$ is not in $G_j$. The case that $d_j(s,t)=2$ is similar.
If $d_i(s,t)=3$, there is a path $(s,y,x,t)$ or $(s,x,y,t)$ in $G_i$. Since $x$ is not in $G_j$ and $y\notin X_j$, recalling that we have removed any $st$-cut node of $G_j$ from $G_i$, the result should also be True. 

The remaining case is $d_i(s,t)=d_j(s,t)=4$. Recall that each graph has at least one $st$-cut node. Any length-4 $st$-path in $G_i$ contains exactly three internal nodes, said  $\{x,y_1,y_2\}$, in which $x\in X_i$ and therefore not in $G_j$. Furthermore neither $y_1$ nor $y_2$ is in $X_j$. Hence, removing the three nodes destroys at most two paths in $G_j$. If there are more than two, not disjoint surely, length-4 $st$-paths in $G_j$, the output  should be True. Similarly it holds if there are more than two such paths in $G_i$. 
The remaining case is that there are one or two paths in either graph, and the answer can be obtained by the following method.
First we choose a path in $G_i$ and check if the removal of the internal nodes separates $s$ and $t$ in $G_j$. If not, we find two disjoint paths. Otherwise we choose the other path in $G_i$ if any, and do it again. 

By the above discussion, the test procedure takes linear time, i.e., $O(|V|+|E_i|+|E_j|)$.
The whole algorithm calls the test procedure for each pair of $i$ and $j$, and therefore the total time complexity is $O(c^2|V|+2c\sum_i |E_i|)=O(c^2n+cm)$ since the other steps of Algorithm~\ref{alg:2path} can be done in $O(cn+m)$ time.
\qed\end{proof}

Combining Algorithms~\ref{alg:inc} and \ref{alg:2path}, we obtain the next theorem. The time complexity is obtained as follows. To implement the while-condition of Algorithms~\ref{alg:inc}, we need to call Algorithm~\ref{alg:2path} at most $|S|$ times, where $|S|$ is the number of paths found so far. Let $p$ be the number of paths found by the algorithm. Since the while-loop may be executed at most $p$ times, the total time complexity is $O(p^2(c^2n+cm))$.
\begin{theorem}
There exists an $O(p^2(c^2n+cm))$ time 2-approximation algorithm for the 4-LCDP problem on a $c$-colors graph, in which $p$ is the number of paths found by the algorithm.
\end{theorem}

Finally we would like to remark the following. In most of the applications, both $c$ and $p$ are small integers, and thus the approximation algorithm runs in linear time. Furthermore, since we need only consider the graphs induced by $\{v|d_i(s,v)+d_i(v,t)\leq 4\}$ for each color $i$, the algorithm is in fact a local algorithm and is therefore efficient even for large-scale social networks. 

\end{document}